\documentclass[twocolumn]{article}

\usepackage{amssymb}
\usepackage{amsmath}
\usepackage{graphicx}
\usepackage[english]{babel}

\usepackage{float}
\restylefloat{table}
\usepackage{booktabs}
\newcommand{\ra}[1]{\renewcommand{\arraystretch}{#1}}
\newcommand{\g}{\mathcal{G}}

\newtheorem{theorem}{Theorem}
\newtheorem{proposition}{Proposition}
\newtheorem{lemma}{Lemma}
\newtheorem{conjecture}{Conjecture}

\newenvironment{proof}{\paragraph{Proof:}}{\hfill$\square$}

\usepackage{authblk}

\title{Some Results on Open Edge and Open Mobile Guarding of Polygons and Triangulations}
\author[1]{Antonio Leslie Bajuelos}
\author[2]{Santiago Canales}
\author[3]{Gregorio Hern\'andez}
\author[1]{Mafalda Martins}
\author[1,4]{In\^es Matos\thanks{ipmatos@ua.pt}}
\affil[1]{Universidade de Aveiro, Portugal}
\affil[2]{Universidad Pontificia Comillas de Madrid, Spain}
\affil[3]{Universidad Polit\'ecnica de Madrid, Spain}
\affil[4]{Universitat Polit\`ecnica de Catalunya, Spain}

\renewcommand\Authands{ and }
\date{}

\begin{document}
 
\maketitle

\begin{abstract}
This paper focuses on a variation of the Art Gallery problem that considers open edge guards and open mobile guards. A mobile guard can be placed on edges and diagonals of a polygon, and the ``open'' prefix means that the endpoints of such edge or diagonal are not taken into account for visibility purposes. This paper studies the number of guards that are sufficient and sometimes necessary to guard some classes of simple polygons for both open edge and open mobile guards. This problem is also considered for planar triangulation graphs using open edge guards.
\end{abstract}

\maketitle

\section{Introduction}

The well known Art Gallery problem studies the minimum number of guards that are needed to fully cover a polygon $P$, that is, the number of guards from which every point of $P$ is visible. Ideally, guards may be placed anywhere on $P$ but usually they are restricted to vertices of the polygon or its edges. In the first case such guards are called vertex guards and in the second edge guards. Moreover, a point guard is a guard that can be placed anywhere on the polygon. Lee et al. proved that finding the minimum number of guards to fully cover a polygon without holes is NP-hard for all three variations of guards \cite{LL86}. Toussaint conjectured that $\lfloor \frac{n}{4} \rfloor$ edge guards are sufficient to cover any simple polygon of $n$ vertices, except for small values of $n$. This exception arose from two examples of simple polygons with $n$ vertices that need $\lfloor \frac{n+1}{4} \rfloor$ edge guards to be fully covered. If guards are able to patrol along the edges and diagonals of $P$ then they are 
called \textit{mobile guards}. In this way, a mobile guard placed on edge $e$ sees a point $p$ of $P$ if some point of $e$ can see $p$. O'Rourke could not prove Toussaint's conjecture, but showed that $\lfloor \frac{n}{4} \rfloor$ mobile guards are sufficient and occasionally necessary to cover any polygon of $n$ vertices \cite{O82}. Later, Shermer proved that $\lfloor \frac{3n}{10} \rfloor$ edge guards are sufficient to cover any simple polygon, except for $n = 3, 6, 13$ where an extra edge guard might be needed \cite{S94}. Shermer actually proved a combinatorial result: any triangulation of a polygon with $n$ vertices can be dominated by $\frac{3n}{10}$ edge guards.

In this paper guards are assumed to be placed along open edges or open diagonals of a polygon, that is, the endpoints of the edge or diagonal are not taken into account for visibility purposes. Therefore, a point $p$ is covered by such guard if $p$ is visible from some interior point of the edge or diagonal. As shown in Figure \ref{pic:OpenEdgeVis}, open edge guards can see considerably less than the usual edge guards and are therefore an interesting topic of research on their own.

\begin{figure}[!htb]
\centering
\includegraphics[scale=0.8]{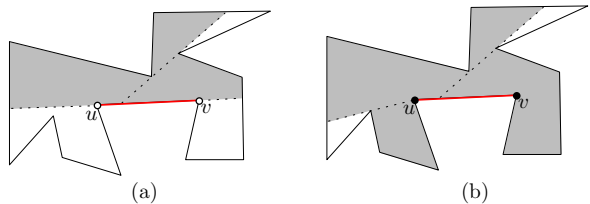}
\caption{(a) The area covered by open edge guard $\overline{uv}$ is shown in grey. (b) The area covered by closed edge guard $\overline{uv}$ is shown in grey.} \label{pic:OpenEdgeVis}
\end{figure}

Open edge guarding is a variation of the Art Gallery problem that was first introduced by Viglietta in 2011, as a way to guard 3D polyhedra, and was published a year later \cite{GioV12}. This work was then built on by Benbernou et al. \cite{BDDK11} and T\'oth et al. \cite{TTW12}, and also by Viglietta himself in his thesis \cite{GV12}. T\'oth et al. studied open edge guards and proved that $\lfloor \frac{n}{3} \rfloor$ guards are necessary to fully cover a simple polygon and $\lfloor \frac{n}{2} \rfloor$ are always sufficient \cite{TTW12}. Following this line of thought, open mobile guards are guards that patrol along open edges and open diagonals of a polygon. Recently, Mukhopadhyay et al. presented yet another variation of open edge guards called the semi-open edge guards \cite{MDT12}, which can monitor every point seen from an interior point of the edge and also from one of its endpoints. They showed that a non star-shaped polygon of $n$ vertices needs at most three semi-open guard edges to be fully 
monitored and proposed an $\mathcal{O}(n)$ algorithm to find all semi-open guard edges of a simple polygon.

This article is divided into two parts: open edge guards and open mobile guards. Section \ref{sec:OpenEdge} is devoted to open edge guards and presents results on the number of guards that cover several types of polygons, such as orthogonal and spiral polygons. Some results related to the Fortress problem on simple and orthogonal polygons are also presented. Section \ref{sec:triangulations} studies open edge guarding of planar triangulation graphs. This problem arises from patrolling triangulated terrains. Section \ref{sec:OpenMobileGuards} introduces open mobile guards and presents results on the number of open mobile guards that fully cover monotone, orthogonal and spiral polygons. Finally, Section \ref{sec:conclusions} concludes the paper and discusses some conjectures, as well as further research.

\section{Open edge guarding of polygons} \label{sec:OpenEdge}

This section studies the problem of calculating the number of open edge guards that are sufficient and sometimes necessary to cover orthogonal polygons, spiral polygons and the exterior of simple polygons. The latter is also called the Fortress problem. Consider the following definition in order to ease the reading of the paper. Given a polygon $P$, $\g_{OE}(P)$ is the minimum number of open edge guards that fully cover $P$ and $\g_{OE}(n) = \min \{ \g_{OE}(P) : \text{$P$ is a polygon of $n$ vertices} \}$. Consequently, this section is devoted to calculate $\g_{OE}(n)$ for different classes of polygons.

\subsection{Orthogonal polygons}\label{sec:Orthogonal}

Bjorling-Sachs proved that $\lfloor \frac{3n+4}{16} \rfloor$ closed edge guards are sufficient and sometimes necessary to fully cover an orthogonal polygon \cite{BS98}. This section shows that $\lfloor \frac{n}{4} \rfloor$ open edge guards are sometimes necessary and always sufficient to fully cover an orthogonal polygon.

Given an orthogonal polygon $P$ with $n$ vertices, the edges of $P$ can be divided into four categories as shown in Figure \ref{pic:OpenEdgeOrtho}(a): north (N), south (S), west (W) and east (E) edges. Each of these four sets represents a group of open edge guards that completely covers $P$. In order to see this, choose a random point $p \in P$. From this point, it is always possible to draw vertical segments that will hit a north edge if it goes up from $p$ and a south edge if it goes down. Similarly, it is always possible to draw horizontal segments through $p$ that will hit a west and an east edge. Therefore, the smallest of these four sets of edges proves the upper bound: any orthogonal polygon can be covered by $\lfloor \frac{n}{4} \rfloor$ open edge guards.

\begin{figure}[!htb]
\centering
\includegraphics[scale=0.8]{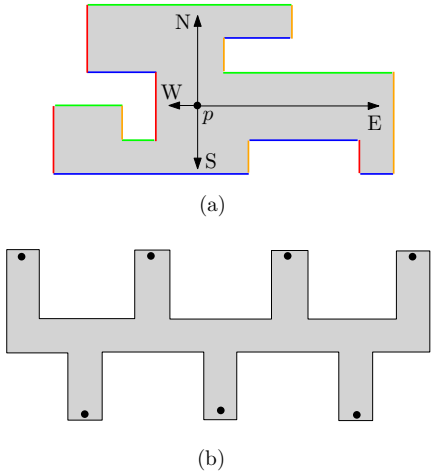}
\caption{(a) Each set of edges drawn in the same trace style fully covers the polygon. (b) Polygon that needs at least one open edge guard to cover each of the marked points.} \label{pic:OpenEdgeOrtho}
\end{figure}

Furthermore, Figure \ref{pic:OpenEdgeOrtho}(b) shows an example of an orthogonal polygon that needs $\lfloor \frac{n}{4} \rfloor$ open edge guards to be fully covered. These two bounds prove the following result.

\begin{theorem}
Any orthogonal polygon of $n$ vertices can be covered by $\lfloor \frac{n}{4} \rfloor$ open edge guards, and in some cases this number is necessary, that is, $\g_{OE}(n) = \lfloor \frac{n}{4} \rfloor$.
\end{theorem}

Observe that this result essentially holds for orthogonal polygons with holes, and the upper bound can be obtained in the same way it was explained above. The lower bound is based on Figure \ref{pic:OpenEdgeOrthoHoles}, which shows an orthogonal polygon with holes that needs $\lfloor \frac{n}{4} \rfloor - 1$ open edge guards to be fully covered, as each marked point is seen by a different open edge guard.

\begin{figure}[!htb]
\centering
\includegraphics[scale=0.7]{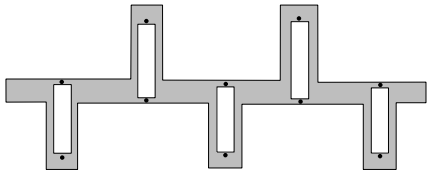}
\caption{A polygon with holes that needs $\lfloor \frac{n}{4} \rfloor - 1$ open edge guards to cover it.} \label{pic:OpenEdgeOrthoHoles}
\end{figure}

\begin{proposition}
Any orthogonal polygon of $n$ vertices with holes can be covered by $\lfloor \frac{n}{4} \rfloor$ open edge guards, and some need $\lfloor \frac{n}{4} \rfloor - 1$ open edge guards in order to be fully covered.
\end{proposition}

\subsection{Spiral polygons}\label{sec:Spirals}

Similarly to what happens with orthogonal polygons, spiral polygons present a set of interesting properties that usually simplify or help turning some visibility problems into more tractable ones. Therefore, this section studies open edge guarding of spiral polygons, which will also be called spirals when it eases the reading of the text. According to a previous work, $\lfloor \frac{n+2}{5} \rfloor$ closed edge guards are sufficient and sometimes necessary to cover spiral polygons \cite{V94}. In the following, it is first shown that $\lfloor \frac{n+1}{4} \rfloor$ open edge guards are always sufficient and occasionally necessary to cover a spiral polygon. Secondly, an algorithm to place the minimum number of open edge guards to cover a spiral polygon is introduced.

\subsubsection{Tight bound on the number of open edge guards}

In the example in Figure \ref{pic:Spiral1}(a), each point marked on the polygon needs a different open edge guard to cover it. Since this spiral only has one possible triangulation and there is one of these points per four triangles, this polygon needs $\lceil \frac{n-2}{4} \rceil$ open edge guards in order to be fully covered. And therefore, $\g_{OE}(n) \geq \lceil \frac{n-2}{4} \rceil$. This lower bound can be rewritten as $\lfloor \frac{n+1}{4} \rfloor$, and below it is proven that this bound is tight: $\g_{OE}(n) = \lfloor \frac{n+1}{4} \rfloor$.

\begin{figure}[!htb]
\centering
\includegraphics[scale=0.7]{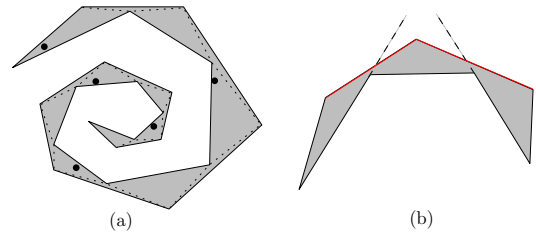}
\caption{(a) No two of the marked points can be covered by the same open edge guard. (b) The two open edge guards marked with a heavier trace cover the whole spiral.} \label{pic:Spiral1}
\end{figure}

The boundary of each spiral can be decomposed in a reflex chain and a convex chain. Let the convex chain be formed by $c$ edges. The reflex chain is formed by $r$ edges and all its vertices are reflex, except for the two endpoints. The proof that any spiral can be covered by $\lfloor \frac{n+1}{4} \rfloor$ open edge guards uses induction on the number of edges of the polygon.

The base case comprehends four cases. When the spiral has four, five or six edges it is easy to see that one open edge guard covers the whole polygon. Spirals with seven edges can be covered by two open edge guards, since it suffices to place the guards on the edges of the convex chain that are intersected by the extensions of the first and last edges of the reflex chain (see Figure \ref{pic:Spiral1}(b)).

\begin{figure}[!htb]
\centering
\includegraphics[scale=0.7]{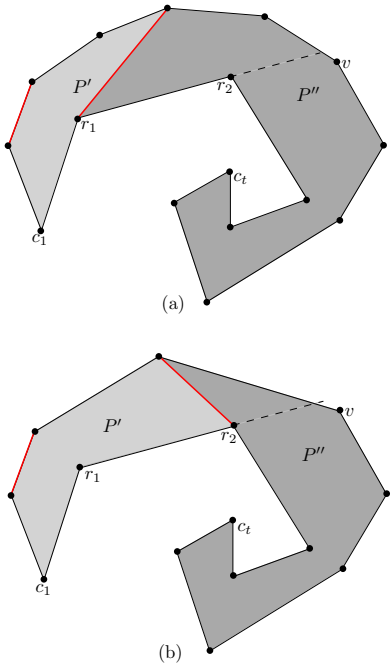}
\caption{(a) The convex chain from $c_1$ to $v$ is formed by six edges. (b) The convex chain from $c_1$ to $v$ is formed by four edges.} \label{pic:Spiral2}
\end{figure}

For the inductive step, suppose $\lfloor \frac{n+1}{4} \rfloor$ open edge guards are sufficient to cover every spiral of $n'$ vertices with $n'<n$ edges, $n > 7$. Let $P$ be a spiral with $n > 7$ vertices, whose reflex chain is formed by the vertices $\{ c_1, r_1, r_2, \ldots, r_k, c_t \}$. Now extend the edge $\overline{r_1r_2}$ until it intersects some edge of the convex chain. Let $v$ be the rightmost endpoint of the convex edge just intersected as shown in Figure \ref{pic:Spiral2}. The proof is now divided into four cases depending on the number of edges of the convex chain from $c_1$ to $v$: (a) five or more than five edges, (b) four edges, (c) three edges and (d) two edges. For case (a), suppose the convex chain from $c_1$ to $v$ has at least five edges (see Figure \ref{pic:Spiral2}(a)). Then draw the diagonal between $r_1$ and $c_5$, the fifth vertex of the convex chain. In this way, the spiral is broken into two spiral polygons: $P'$ that has six edges and so can be guarded with one open edge guard 
and $P''$ with $n-4$ edges. For case (b), suppose the convex chain from $c_1$ to $v$ has four edges as shown in Figure \ref{pic:Spiral2}(b). Then draw the diagonal between $r_2$ and $c_4$, the fourth vertex of the convex chain. This breaks the spiral in two: $P'$ has six edges and therefore can be guarded with one open edge guard and $P''$ with $n-4$ edges.

In case (c) the convex chain from $c_1$ to $v$ has three edges and the situation is slightly different. As shown in Figure \ref{pic:Spiral3}, draw the diagonal between $v$ and the first visible reflex vertex starting from $r_2$ (note that $r_2$ can be such vertex). This procedure breaks the spiral in two polygons: $P'$ that can be guarded with one open edge guard and $P''$ with at most $n-4$ edges.%Note that $P'$ can have a lot of edges, but most of these are on the reflex chain since the convex has only three. Therefore, the middle edge of the convex chain can cover the whole polygon.

\begin{figure}[!htb]
\centering
\includegraphics[scale=0.7]{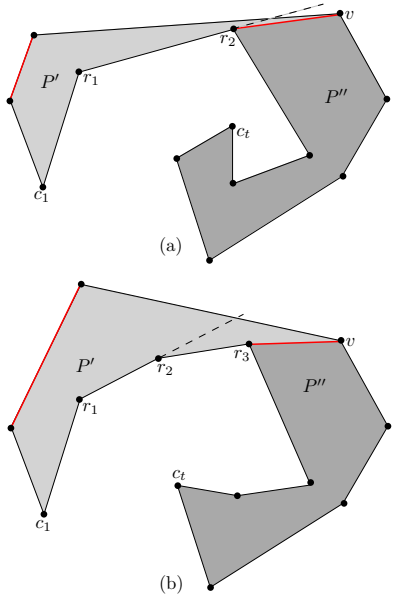}
\caption{The convex chain from $c_1$ to $v$ is formed by three edges. (a) The first reflex vertex visible from $v$ is $r_2$. (b) The first reflex vertex visible from $v$ is $r_3$.} \label{pic:Spiral3}
\end{figure}

Finally, case (d) in which the convex chain from $c_1$ to $v$ only has two edges. In this case, draw the diagonal from $v$ to the first visible reflex vertex after $r_2$. If there are no visible reflex vertices left, then the reflex chain is over and an open edge guard placed on the second edge of the convex chain covers the whole spiral (see Figure \ref{pic:Spiral4}(a)). If there is one visible reflex vertex then draw the diagonal as before, which will break the spiral in two polygons: $P'$ that can be guarded with one open edge guard and $P''$ with at most $n-4$ edges (see Figure \ref{pic:Spiral4}(b)). 

\begin{figure}[!htb]
\centering
\includegraphics[scale=0.7]{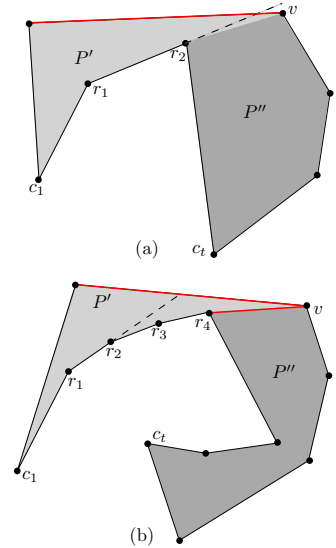}
\caption{The convex chain from $c_1$ to $v$ is formed by two edges. (a) There is no reflex vertex visible from $v$ besides $r_2$. (b) The first reflex vertex visible from $v$ is $r_4$.} \label{pic:Spiral4}
\end{figure}

All the four cases described above end with a polygon $P''$ that has at most $n-4$ edges, which means the inductive hypothesis can be applied. Therefore, $P''$ can be covered by $\lfloor \frac{(n-4)+1}{4} \rfloor = \lfloor \frac{n+1}{4} \rfloor - 1$ open edge guards. Since polygon $P'$ is covered exactly by one open edge guard, the whole spiral is covered by $\lfloor \frac{n+1}{4} \rfloor$ open edge guards and this concludes the proof.

\begin{theorem}
Any spiral polygon with $n$ vertices can be covered by $\lfloor \frac{n+1}{4} \rfloor$ open edge guards, and in some cases this number is necessary.
\end{theorem}

If the spiral polygon is also orthogonal, then the previous result can be improved.

\begin{proposition}
Every orthogonal spiral polygon with $n$ vertices can be covered by $\lceil \frac{n-2}{6} \rceil$ open edge guards. This bound is tight for all polygons of this class.
\end{proposition}

\begin{proof}
Any such polygon is only covered if an open edge guard is placed on the middle edge of every group of three consecutive edges of the convex chain. Therefore, if $c$ is the number of edges on the convex chain, the final number of open edge guards is $\lceil \frac{c}{3} \rceil = \lceil \frac{n-2}{6} \rceil$.
\end{proof}

\subsubsection{Placing the minimum number of open edge guards}

This section presents an algorithm to place the minimum number of open edge guards that cover a spiral polygon $P$. The main idea of the algorithm is to build two sets simultaneously: $G$, which is the set of open edge guards, and $H$ that is the set of points that guarantees $G$ is of minimum size. The points that form set $H$ are placed on the polygon in such a way that each open edge guard can only see one of them and is therefore associated with it. Consequently, $|G| = |H|$. Let $\{r_1, r_2, \ldots, r_k \}$ be the set of reflex vertices and $\{c_1, c_2, \ldots, c_{n-k} \}$ the set of convex vertices of $P$. Moreover, let $\{c_1, c_2, \ldots, c_{n-k}, r_k, r_{k-1}, \ldots, r_1 \}$ be the sequence of $n$ vertices of spiral polygon $P$. The stpdf of the algorithm to place the minimum number of open edge guards to cover $P$ are depicted in Figure \ref{pic:AlgSpiral1} and detailed in the following.

\begin{figure*}[!htb]
\centering
\includegraphics[scale=0.7]{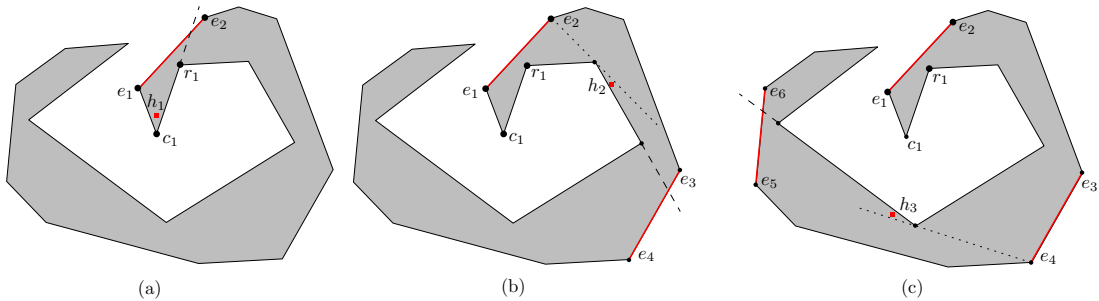}
\caption{(a) Finding an edge of the convex chain that covers point $h_1$. (b) Finding an edge of the convex chain that covers point $h_2$. (c) Polygon $P$ is fully covered by the three open edge guards.} \label{pic:AlgSpiral1}
\end{figure*}

\begin{enumerate}
 
\item $G \leftarrow \emptyset$;

\item Let $h_1$ be a point very close to $c_1$, which has to be covered. Draw the ray $\overrightarrow{c_1r_1}$ that will intersect some edge of the convex chain that sees point $h_1$. Let such edge be denoted by $(e_1,e_2)$ and assign $G \leftarrow \{ (e_1,e_2) \}$.

\item Find the last reflex vertex $r_j$ that can be seen from $e_2$ and consider point $h_2$, which is very close to $r_j$ along the edge $(r_j,r_{j+1})$. Draw the ray $\overrightarrow{r_jr_{j+1}}$ that will intersect some edge of the convex chain that sees point $h_2$. Let such edge be denoted by $(e_3,e_4)$ and assign $G \leftarrow G \cup \{ (e_3,e_4) \}$.

\item Repeat the last step until all reflex vertices and $c_{n-k}$ are guarded.

\end{enumerate}

This algorithm selects the convex edges $\overline{e_j e_{j+1}}$ that will be part of set $G$, which fully covers any spiral polygon since it totally covers its convex chain \cite{NW90}. The idea is to associate $h_1$ to the edge starting at $e_1$, which is the last vertex of the convex chain that sees $h_1$. Then this procedure is repeated for point $h_2$, which is placed close to the first reflex vertex that is not seen by $\overline{e_1 e_2}$. Afterwards $\overline{e_3 e_4}$ is selected as the last edge of the convex chain that sees $h_2$, and so on until every reflex vertex is covered. It is left to prove that the algorithm described above does indeed place the minimum number of open edge guards needed to cover a spiral polygon.

\begin{lemma} \label{lem:NumbGuards}
The algorithm described above builds a set $H$ of points interior to $P$ such that $\g_{OE}(P) \geq |H|$.
\end{lemma}

\begin{proof}
As the algorithm runs, several points $h_i$ are placed in a way that hides them from the edges of the convex chain of $P$ that were chosen as open edge guards. Let $H$ be the set of all these points. Note that all points of $H$ are visibly independent, that is, if $h_i \neq h_j$ then there is no open edge guard that can cover both points, and therefore $\g_{OE}(P) \geq |H|$.
\end{proof}

\begin{theorem}
The algorithm described in this section places the minimum number of open edge guards needed to cover a spiral polygon in $\mathcal{O}(n)$ time.
\end{theorem}

\begin{proof}
Let $G$ be the set of open edge guards chosen by the algorithm to cover spiral polygon $P$ and let $H$ be the set of hidden points. According to Lemma \ref{lem:NumbGuards}, $\g_{OE}(P) \geq |H|$ but since $|H|=|G|$ then $\g_{OE}(P) \geq |G|$ and therefore $G$ is a minimum set of open edge guards. Regarding the time complexity, each edge of the convex chain is only processed once whilst analysing the rays $\overrightarrow{r_jr_{j+1}}$. In the same way, each edge of the reflex chain is checked once to find the last reflex vertex that is visible from the chosen edges. Consequently, each vertex of the spiral polygon is analysed just once by the algorithm and therefore it runs in linear time.
\end{proof}

\subsection{Fortress problem}

This section is devoted to another variation of the Art Gallery problem called the \textit{Fortress Problem}. Instead of guarding the interior of a simple polygon, the Fortress Problem variation focuses on monitoring the exterior of a polygon. This problem has been studied for both vertex and edge guards, but the results below are naturally be associated with open edge guards. Choi et al. proved that the exterior of any simple polygon can be covered by $\lceil \frac{n}{3} \rceil$ edge guards and that these guards are necessary to cover the exterior of convex polygons \cite{CY01}. In the case of open edge guards, this problem is trivial since it is easy to see that every edge will be needed as a guard to cover the exterior of a convex polygon. Consequently, the exterior of any simple polygon with $n$ vertices can be covered by $n$ open edge guards, and in some cases this number is necessary.

The natural following step is to study orthogonal polygons. Again, Choi et al. proved that the exterior of any orthogonal polygon can be covered by $\lfloor \frac{n}{4} \rfloor +1$ edge guards and that this number can be necessary \cite{CY01}.

The proof of the following theorem is based on the technique of dividing the edges according to their orientation, as introduced in Section \ref{sec:Orthogonal}. As Figure \ref{pic:Fortress1}(a) shows, the edges of an orthogonal polygon can be divided into four groups: north edges (N), south edges (S), west edges (W) and east edges (E). Let $R$ be the smallest rectangle that encloses polygon $P$. Every point outside $P$ and within $R$ can be covered by open edges of type N or S. This can be easily seen as a vertical line through point $p \in R \backslash P$ will always intersect a north edge or a south edge (or both). The same happens for edges of type E and W and it is fairly easy to realise that there are as many edges N and S as E and W. If there is an open edge guard on every edge of type N or S, then the whole region $R \backslash P$ is covered by $\frac{n-4}{2}$ open edge guards. The exterior of $R$ can be guarded by four open edge guards, each placed on the extreme edges that dominate $R$ (topmost, 
bottommost, leftmost and rightmost edges). Therefore, the exterior of $P$ is covered by $\frac{n-4}{2} + 4 = \frac{n}{2} + 2$ open edge guards. 

Finally, the lower bound is given by the orthoconvex polygon depicted in Figure \ref{pic:Fortress1}(b). Each of the marked points is covered by a different open edge guard, and so $\frac{n}{2} + 2$ is a lower bound for this problem.

\begin{figure}[!htb]
\centering
\includegraphics[scale=0.8]{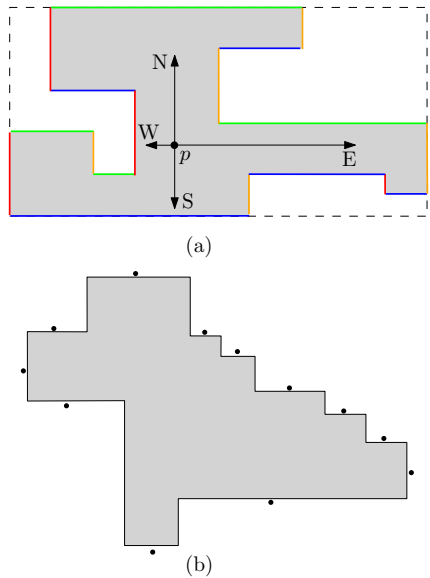}
\caption{(a) The edges of the polygon are divided into N, S, W and E edges. (b) An orthoconvex polygon that needs $\frac{n}{2} + 2$ open edge guards to be fully covered.} \label{pic:Fortress1}
\end{figure}

\begin{theorem}
The exterior of any orthogonal polygon with $n$ vertices can be covered by $\frac{n}{2}+2$ open edge guards, and in some cases this number is necessary.
\end{theorem}

This result can also be rewritten using reflex vertices. Since $r +2 = \frac{n}{2}$, $r+4$ open edge guards are always sufficient and occasionally necessary to fully cover the exterior of an orthogonal polygon.

\section{Open edge guarding of triangulations} \label{sec:triangulations}

In the '80s some of the research in visibility problems shifted to polyhedral terrains. A terrain is a polyhedral surface whose intersection with a vertical line is either empty or a point. Terrains are often considered to be triangulated in such way that all its faces are triangles. If one of these terrains is orthogonally projected onto a plane below, it becomes a planar triangulation graph, that is, the graph of a triangulation of a set of points on the plane. This allows a combinatorial correspondence between guarding a triangulated terrain and guarding its projection. A set of edges $H$ is called a set of guards for a triangulation if each face of such triangulation is covered by at least one guard of $H$.

Bose et al. showed that the problem of finding the minimum number of edge guards that fully cover a terrain is NP-complete \cite{BSTB97}. Even though it was published later, Everett and Rivera-Campo proved in 1994 that $\lfloor \frac{n}{3} \rfloor$ edge guards suffice to cover every face of a planar triangulation with $n$ vertices \cite{ERC97}. The lower bound was established by Batista et al. using a triangulation that needs $\lfloor \frac{n-2}{3} \rfloor$ edge guards in order to be fully covered \cite{BPR10}.

Naturally, this section studies the variation of this problem that concerns open edge guards.

\subsection{Planar triangulations}

Since this section deals with open edges, a guard placed on an interior edge of a triangulation is only able to patrol the two triangles incident to such edge (see Figure \ref{pic:OpenEdgeTriang1}). The problem of finding a set of open edge guards that covers a triangulation $G$ can be translated to finding an edge cover of the dual graph $G^*$. If $\g_{OE}(G)$ is the size of a minimum set of open edge guards that covers $G$, then $\g_{OE}(G) = \beta'(G^*)$, where $\beta'(G^*)$ is the size of a minimum edge cover of $G^*$, that is, the edge covering number.

\begin{figure}[!htb]
\centering
\includegraphics[scale=0.7]{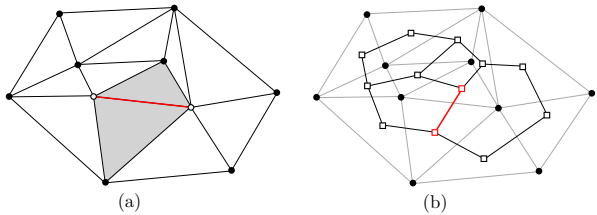}
\caption{(a) An open edge guard placed on the bold edge covers the two grey triangles. (b) The bold edge becomes an edge of the dual graph uniting the nodes representing the grey triangles.} \label{pic:OpenEdgeTriang1}
\end{figure}

In the following it is shown how to calculate an upper bound for $\beta'(G^*)$. As Figure \ref{pic:OpenEdgeTriang1} shows, each interior vertex of $G$ becomes a cycle without chords in $G^*$. If an edge $e \in G^*$ is deleted from one of such cycles then $\beta'(G^*) \leq \beta'(G^* \backslash e)$, that is, there is one less edge but the number of nodes to cover remains the same. Removing these edges breaks the cycles of $G^*$, and once there are no cycles left the graph becomes a spanning tree $T^*$ of maximum degree three. And consequently, $\beta'(G^*) \leq \beta'(T^*)$. Consider the following theorem.

\begin{theorem}
Every tree $T^*$ with $t$ nodes and maximum degree three allows an edge cover that has $\lfloor \frac{2t+1}{3} \rfloor$ edges at most, that is, $\beta'(T^*) \leq \lfloor \frac{2t+1}{3} \rfloor$.
\end{theorem}

\begin{proof}
This proof takes advantage of the relation between edge covering and matching. The latter is a subset of edges of a graph without common vertices. Let $\alpha'(T^*)$ be the maximum number of edges of a matching in $T^*$. According to Gallai's theorem, if a graph $G$ has $m$ nodes and none of them is isolated, then $\alpha'(G) + \beta'(G) = m$. Therefore, it suffices to prove that there is a matching with $t - \lfloor \frac{2t+1}{3} \rfloor$ edges in every tree $T^*$ with $t$ nodes and maximum degree three. Since $t - \lfloor \frac{2t+1}{3} \rfloor = \lceil \frac{t-1}{3} \rceil$, it is enough to show that there is a matching in $T^*$ with $\lceil \frac{t-1}{3} \rceil$ edges. A vertex cover of a graph is a set of vertices such that each edge of the graph is incident to at least one vertex of the set. The minimum number of vertices in a vertex cover of a graph $G$ is denoted by $\beta(G)$. According to K\"{o}nig's theorem, this number coincides with $\alpha'(G)$ when $G$ is a bipartite graph. The size of a 
matching in a bipartite graph is given by the following lemma.

\begin{lemma}
There is a matching with at least $\frac{q}{\Delta}$ edges in every bipartite graph of $q$ edges and maximum degree $\Delta$.
\end{lemma}

\begin{proof}
To prove this lemma, observe that each vertex of graph $G$ only covers $\Delta$ edges at most. Therefore, at least $\frac{q}{\Delta}$ vertices are needed to cover all edges of $G$ and so $\beta(G) \geq \frac{q}{\Delta}$. Since $G$ is a bipartite graph, $\alpha'(G) = \beta(G) \geq \frac{q}{\Delta}$ and consequently there is a matching in $G$ with at least $\frac{q}{\Delta}$ edges.
\end{proof}

To conclude the proof of the theorem, note that tree $T^*$ is a bipartite graph with $t-1$ edges and maximum degree three. Therefore, and according to the previous lemma, $\alpha'(T^*) \geq \lceil \frac{t-1}{3} \rceil$ and so $\g_{OE}(G) = \beta'(G^*) \leq \beta'(T^*) \leq \lfloor \frac{2t+1}{3} \rfloor$.
\end{proof}

If $\g_{OE}(t)$ is defined as $\g_{OE}(t) = \max \{ \g_{OE}(G) : \text{$G$ is a triangulation with }$ $\text{$t$ triangles} \}$ then $\g_{OE}(t)$ is bounded from above by $\lfloor \frac{2t+1}{3} \rfloor$, that is, $\g_{OE}(t) \leq \lfloor \frac{2t+1}{3} \rfloor$.

Finally, it is left to prove that there are triangulations that need $\lfloor \frac{2t+1}{3} \rfloor$ open edge guards in order to be fully covered. Such triangulation springs from a tree $T^*$ with $t$ nodes and maximum degree three. This tree is built on a path of even length by adding a leaf connected to every other node (see Figure \ref{pic:OpenEdgeTriang2}(a)).

\begin{figure}[!htb]
\centering
\includegraphics[scale=0.7]{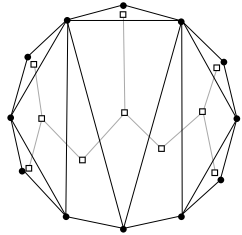}
\caption{(a) Tree $T^*$ in which there is a leaf connected to every other node. (b) A triangulation that needs $\lfloor \frac{2t+1}{3} \rfloor$ open edge guards in order to be fully covered.} \label{pic:OpenEdgeTriang2}
\end{figure}

Let $H$ be the set of nodes of $T^*$ of degree one and two. Each of these nodes is covered by a different edge of $T^*$, and so if $C$ is an edge cover of $T^*$ then $|C| \geq |H|$. To count the number of elements of $H$, observe that if a tree has $t_1$ nodes of degree one, $t_2$ nodes of degree two and $t_3$ nodes of degree three, then $t_1+2t_2+3t_3 = 2(t_1+t_2+t_3-1)$ and therefore $t_1 = 2+ t_3$. In the example in Figure \ref{pic:OpenEdgeTriang2}(a) one can observe that $t_2 = t_3 -1$, and so $t = t_1+t_2+t_3 = 3t_3+1$. The number of elements of $H$ is then given by: 

$$ |H| = t_1 + t_2 = 2t_3+1 = 2 \frac{t-1}{3} + 1 = \frac{2t+1}{3} $$

According to this, any edge cover of $T^*$ has to be formed by at least $\frac{2t+1}{3}$ edges. In the example, the dual-graph $T^*$ has $t$ nodes where $t \equiv 1 \pmod{3}$. Note that if a node is added to $T^*$ (adjacent to one of its leaves) or two nodes are added (adjacent to different leaves), the minimum number of edges of an edge cover of $T^*$ does not change. This observation concludes the proof of the lower bound, as for any value of $t>1$ there are examples of trees that are dominated by sets formed by at least $\frac{2t+1}{3}$ edges. Figure \ref{pic:OpenEdgeTriang2}(b) shows a triangulation whose dual-graph is tree $T^*$ described above.

The two bounds achieved above prove the following theorem.

\begin{theorem}
Any triangulation with $t$ triangles can be covered by $\lfloor \frac{2t+1}{3} \rfloor$ open edge guards and in some cases this number is necessary, that is, $\g_{OE}(t) = \lfloor \frac{2t+1}{3} \rfloor$.
\end{theorem}

\section{Open mobile guarding of polygons} \label{sec:OpenMobileGuards}

Since open edge guarding is difficult to tackle, this section allows open edge guards to monitor diagonals of the polygons as well, and therefore it studies open mobile guards. Given a polygon $P$, $\g_{OM}(P)$ is the minimum number of open mobile guards that fully cover $P$, that is, $\g_{OM}(n) = \min \{ \g_{OM}(P) : \text{$P$ is a polygon of $n$ vertices} \}$.

\subsection{Monotone polygons}\label{sec:Monotone}

Since open mobile guards can also be placed on both edges and diagonals of a polygon, this section uses the term open edge guard to make it clearer when it suffices to place a guard on an edge and no diagonals are needed.

\begin{theorem}
Any monotone polygon with $n$ vertices can be covered by $\lfloor \frac{n}{3} \rfloor$ open mobile guards. This bound is tight for all polygons of this class.
\end{theorem}

The necessity of $\lfloor \frac{n}{3} \rfloor$ open mobile guards to cover a monotone polygon is given by the example in Figure \ref{pic:Monotone1}. There is no open mobile guard that can see two of the points ``hidden'' inside the spikes. The number of guards follows directly from this example as there is one of these points for every three vertices of $P$. Observe that this example also holds as a lower bound for open edge guards.

\begin{figure}[!htb]
\centering
\includegraphics[scale=0.7]{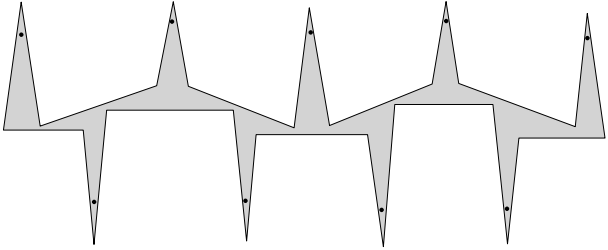}
\caption{No two black points are seen by the same open mobile guard.} \label{pic:Monotone1}
\end{figure}

The sufficiency proof is attained by induction on $n$. This proof is immediate for $n=3,4,5$ since one open mobile guard suffices to fully cover the whole monotone polygon. If $n=6$, then two open mobile guards can be necessary. For the inductive step, suppose that $\lfloor \frac{n}{3} \rfloor$ open mobile guards are sufficient to any cover polygon of $n'<n$ vertices, $n>6$. Let $P$ be a polygon of $n$ vertices that is monotone with respect to a horizontal line. If the vertices of $P$ are sorted from left to right, let $r$ be a vertical line between vertices $3$ and $4$. Supposing vertex $4$ is on the bottom chain, let $u$ be the first vertex on the right of $r$ on the top chain of $P$. There are two possible cases depending on whether the diagonal $\overline{4u}$ exists (see Figure \ref{pic:Monotone2}). If the diagonal $\overline{4u}$ exists on $P$ then the polygon is divided into two monotone polygons: $P'$ of five vertices that can be covered using one open mobile guard and $P''$ of $n-3$ vertices (see 
Figure \ref{pic:Monotone2}(a)). By the inductive hypothesis, $P''$ can be covered by $\lfloor \frac{n-3}{3} \rfloor = \lfloor \frac{n}{3} \rfloor - 1$ open mobile guards. Consequently, $\lfloor \frac{n}{3} \rfloor$ open mobile guards suffice to cover the whole polygon $P$.

\begin{figure}[!htb]
\centering
\includegraphics[scale=0.8]{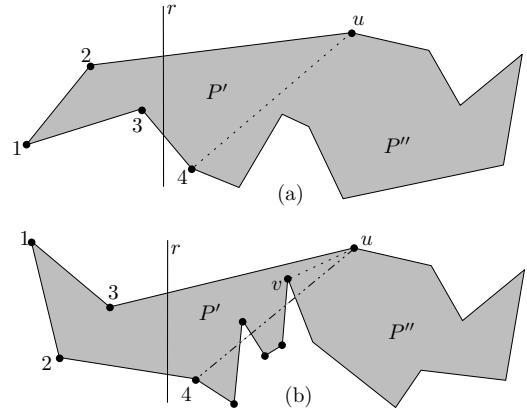}
\caption{(a) The diagonal $\overline{4u}$ exists on polygon $P$ and divides it into polygons $P'$ and $P''$. (b) The diagonal $\overline{4u}$ does not exist on $P$ but $\overline{uv}$ does and it divides $P$ into two monotones polygons.} \label{pic:Monotone2}
\end{figure}

If the diagonal $\overline{4u}$ does not exist then let $v$ be the last vertex on the bottom chain on the left of $\overline{4u}$, which is visible to $u$ (as shown in Figure \ref{pic:Monotone2}(b)). Consequently, $\overline{uv}$ is a diagonal that exists on $P$ and divides the polygon into two monotone ones: $P'$ and $P''$. There are three cases to be considered in this situation: $P''$ has at most $n-6$ edges, $P''$ has $n-4$ edges or $P''$ has $n-5$ edges. If $P''$ has at most $n-6$ edges, then $P''$ is covered by $\lfloor \frac{n-6}{3} \rfloor = \lfloor \frac{n}{3} \rfloor - 2$ open mobile guards according to the inductive hypothesis. And to conclude, two extra open mobile guards are needed to cover $P'$: one is placed on $\overline{3u}$ and the other on the edge connecting vertices $2$ and $4$ (see Figure \ref{pic:Monotone2}(b)). If $P''$ has $n-4$ edges then $v$ is the fifth vertex of $P$ (see Figure \ref{pic:Monotone3}(a)). Now divide $P$ using the diagonal $\overline{3v}$, which creates the pentagon 
$P^*$ that can be fully covered by one open mobile guard. The polygon yet to be covered has $n-3$ edges since it is formed by the union of $P''$ and the triangle formed by the vertices $3$, $u$ and $v$. According to the inductive hypothesis, this polygon can be covered by $\lfloor \frac{n-4}{3} \rfloor \leq \lfloor \frac{n}{3} \rfloor - 1$ open mobile guards. However, vertex $v$ may not see any other vertex of the top chain on the left of $u$. If edge $\overline{1u}$ exists then it fully covers $P'$. If it does not, then $P$ has to be split using diagonal $\overline{uv}$ as shown in Figure \ref{pic:Monotone3}(b). In this case, $P'$ is covered by placing a guard on diagonal $\overline{3u}$ if vertex $2$ is reflex or on edge $\overline{2u}$ otherwise. Note that this is the only case of this proof where an open mobile guard is needed. To conclude this case, polygon $P''$ has $n-4$ edges and therefore can be guarded by $\lfloor \frac{n}{3} \rfloor - 1$ open mobile guards.

\begin{figure}[!htb]
\centering
\includegraphics[scale=0.8]{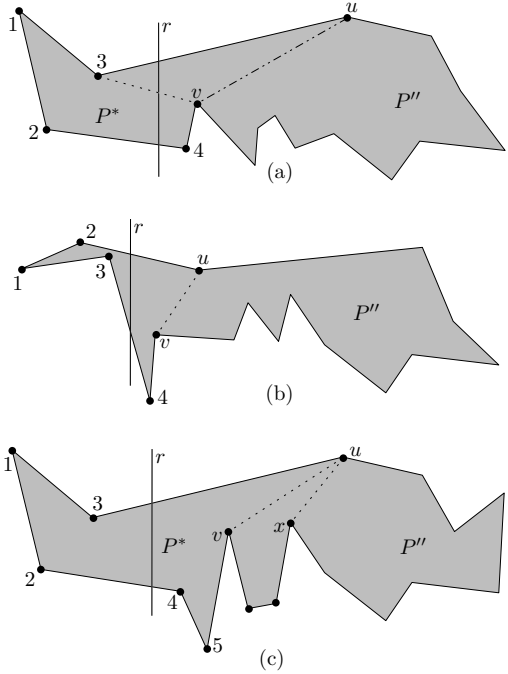}
\caption{(a) The diagonal $\overline{uv}$ divides $P$ into two polygons and $v$ is the fifth vertex. (b) Vertex $v$ does not see any other vertex of the top chain on the left of $u$ and $\overline{1u}$ does not exist. (c) The diagonal $\overline{uv}$ divides $P$ into two polygons and $v$ is the sixth vertex.} \label{pic:Monotone3}
\end{figure}

Finally, if $P''$ has $n-5$ edges then $v$ is the sixth vertex of $P$. In this case, consider the diagonal $\overline{ux}$ where $x$ is the first vertex of the bottom chain after $v$ that is visible to $u$ (see Figure \ref{pic:Monotone3}(c)). Diagonal $\overline{ux}$ divides $P$ into two monotone polygons, one of them is $P^*$ which can be fully covered using two open mobile guards placed on edges $\overline{3u}$ and $\overline{24}$. The piece of the polygon $P$ yet to cover has $n-6$ vertices, and consequently can be covered by $\lfloor \frac{n-6}{3} \rfloor = \lfloor \frac{n}{3} \rfloor - 2$ open mobile guards. This concludes the proof that $\lfloor \frac{n}{3} \rfloor$ open mobile guards suffice to fully cover any monotone polygon.

Note that in the previous examples vertices $2$ and $3$ lie in different chains. If both vertices are in the same chain, the proof still holds but it is possible that the open mobile guard placed on edge $\overline{24}$ that is covering $P'$ or $P^*$ needs to be swapped for the correct edge.

As previously mentioned, there is only one case of this proof where the diagonals of the polygon are needed and so we believe this result holds for open edge guards. Recall that the example in Figure \ref{pic:Monotone1} already proves the lower bound.

\begin{conjecture}
Any monotone polygon with $n$ vertices can be covered by $\lfloor \frac{n}{3} \rfloor$ open edge guards.
\end{conjecture}

\subsection{Orthogonal polygons}\label{sec:MobOrthogonal}

In Section \ref{sec:Orthogonal}, Figure \ref{pic:OpenEdgeOrtho}(b) shows an orthogonal polygon that needs $\lfloor \frac{n}{4} \rfloor$ open edge guards in order to be fully covered. This example does not hold for open mobile guards as a guard placed on some diagonals of the polygon is able to see two of the marked points. However, the orthogonal polygon shown in Figure \ref{pic:MobOrthogonal} can be used to prove the lower bound for this type of guards: at least $\lfloor \frac{n+1}{5} \rfloor$ open mobile guards are needed to cover such polygon. In this example, each set of ten vertices is only covered if there are at least two open mobile guards.

\begin{figure}[!htb]
\centering
\includegraphics[scale=0.8]{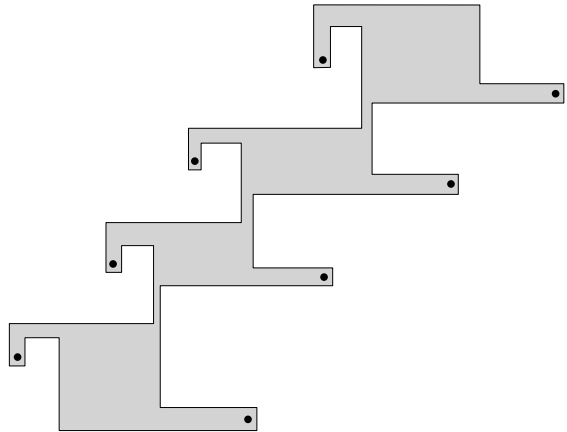}
\caption{Orthogonal polygon that needs at least one open mobile guard to cover each of the marked points.} \label{pic:MobOrthogonal}
\end{figure}

Still, this is a short example and the number of vertices is a multiple of ten. Nevertheless, this polygon can be generalised to $n$ in order to achieve the following result.

\begin{proposition}
There are orthogonal polygons with $n$ vertices that need at least $\lfloor \frac{n+1}{5} \rfloor$ open mobile guards to be fully covered, that is, $\g_{OM}(n) \geq \lfloor \frac{n+1}{5} \rfloor$.
\end{proposition}

Obviously, we also believe this lower bound coincides with the upper bound.

\begin{conjecture}
Any orthogonal polygon with $n$ vertices can be covered by $\lfloor \frac{n+1}{5} \rfloor$ open mobile guards.
\end{conjecture}

\subsection{Spiral polygons}\label{sec:MobSpirals}

\begin{table*}[!htb]
 \centering
  \ra{1.4}
   \begin{tabular}{@{}lcc|cc@{}} \toprule
        Class of Polygons & \multicolumn{2}{c}{Open Edge Guards} & \multicolumn{2}{r}{Open Mobile Guards} \\
	\hline	
	Orthogonal && $\g_{OE}(n) = \lfloor \frac{n}{4} \rfloor$ && $\g_{OM}(n) \geq \lfloor \frac{n+1}{5} \rfloor$ \\
	Orthogonal with holes && $\lfloor \frac{n-4}{4} \rfloor \leq \g_{OE}(n) \leq \lfloor \frac{n}{4} \rfloor$ && - \\
	Spirals && $\g_{OE}(n) = \lfloor \frac{n+1}{4} \rfloor$ && $\g_{OM}(n) = \lfloor \frac{n+1}{4} \rfloor$ \\
	Orthogonal Spirals && $\g_{OE}(n) = \lceil \frac{n-2}{6} \rceil$ && $\g_{OM}(n) = \lceil \frac{n-2}{6} \rceil$ \\
	Monotone  && $\g_{OE}(n) \geq \lfloor \frac{n}{3} \rfloor$ && $\g_{OM}(n) = \lfloor \frac{n}{3} \rfloor$ \\
        \bottomrule
        \bottomrule
        Planar Triangulations && $\g_{OE}(t) = \lfloor \frac{2t+1}{3} \rfloor$ && - \\
        \bottomrule
    \end{tabular}
    \caption{Summary of the results.} \label{table}
\end{table*}

Even though open edge guards are more restrictive than open mobile guards, the spiral polygon depicted in Figure \ref{pic:Spiral1}(a) in Section \ref{sec:Spirals} also works as a lower bound concerning open mobile guards. Furthermore, the proof of the upper bound presented in that section also holds for open mobile guards. Consequently, the following theorem follows directly from the results proven above. 

\begin{theorem}
Any spiral polygon with $n$ vertices can be covered by $\lfloor \frac{n+1}{4} \rfloor$ open mobile guards, and in some cases this number is necessary.
\end{theorem}

\section{Conclusions and further research} \label{sec:conclusions}

This paper introduced several results on the number of both open edge and open mobile guards concerning the coverage of some classes of simple polygons. These results are summarised in Table \ref{table}. Although omitted from the table, the Fortress Problem -- guarding the exterior of a polygon -- was also studied and some results were obtained, mainly for orthogonal polygons. We believe that finding the minimum number of open edge guards to cover a polygon is NP-hard, as this is the case for closed edge guards \cite{LL86}. Notwithstanding, Section \ref{sec:Spirals} introduced an algorithm to place the minimum number of open edge guards to cover a spiral polygon. Further research on this topic would focus on finding an algorithm to place the guards in the case of the fortress problem. We also believe that the bound of $\lfloor \frac{n}{3} \rfloor$ open mobile guards to cover monotone polygons is a tight bound for open edge guards as well. On open mobile guards, tighten the bound of $\lfloor \frac{n+1}{5} \rfloor$ for orthogonal polygons remains a future goal.

Section \ref{sec:triangulations} introduced the problem of covering planar triangulations using open edge guards. The main theorem of the section can also be found in Table \ref{table} and states that for any triangulation with $t$ triangles, $\g_{OE}(t) = \lfloor \frac{2t+1}{3} \rfloor$. As further research, it would be interesting to find an efficient algorithm to calculate the minimum number of open edge guards that cover a given triangulation $G$. Micali et al. showed that it is possible to build a matching of maximum size in a graph of $n$ nodes and $m$ edges in $\mathcal{O}(\sqrt{n} m)$ time \cite{MV80}. Therefore, there is a polynomial algorithm to calculate $\g_{OE}(G)$, but it would be compelling to improve this result.

To conclude, we presented several results for this variation of the Art Gallery problem, but we also opened a series of unsolved problems and interesting conjectures to tackle in the future. \vspace*{0.5cm}

\subsection*{Acknowledgements}
The first, fourth and fifth authors research was supported by {\it FEDER} funds through {\it COMPETE}--Operational Programme Factors of Competitiveness, CIDMA and FCT within project PEst-C/MAT/UI4106/2011 with COMPETE number FCOMP-01-0124-FEDER-022690. The third author was supported by ESF EUROCORES programme EuroGIGA - ComPoSe IP04 - MICINN Project EUI-EURC-2011-4306. The fourth and fifth authors were also supported by FCT grants SFRH/BPD/66431/2009 and  SFRH/BPD/66572/2009, respectively.

\bibliographystyle{abbrv}
\bibliography{OpenEdgeGuards}

\end{document}